%% file: article.tex
\documentclass[runningheads]{llncs}

\input{preamble}
\input{macros}

\input{solidity-highlighting.tex}

\title{Transaction Monitoring of Smart Contracts}

\def\orcid#1{\smash{\href{http://orcid.org/#1}{\protect\raisebox{-1.25pt}{\protect\includegraphics{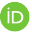}}}}}

\author{Margarita Capretto\inst{1,2}\orcid{0000-0003-2329-3769}\and
  Martin Ceresa\inst{1}\orcid{0000-0003-4691-5831} \and
  C\'esar S\'anchez\inst{1}\orcid{0000-0003-3927-4773}
}
\authorrunning{Capretto, Ceresa and Sánchez}
\institute{IMDEA Software Institute, Spain \and
  Universidad Politécnica de Madrid (UPM), Madrid, Spain}

\begin{document}

\maketitle

\input{Abstract}


\input{Introduction}


\input{ModelOfComputation}
\input{TransactionalMonitors}
\input{ExecutionMechanisms}

\input{results}


\input{Discussion}

\vfill\pagebreak

\bibliographystyle{abbrv}
\bibliography{bibfile}

\vfill\pagebreak

\newcounter{backup}
\appendix
\input{Appendix/Mechanisms.Equiv.tex}

\vfill\pagebreak

\input{Appendix/BFS.ImpProof}

\vfill\pagebreak

\input{Appendix/DFS.Proofs}

\vfill\pagebreak

\input{Appendix/FlashLoans.tex}
\end{document}

%% file: preamble.tex
\usepackage[utf8]{inputenc}
\usepackage[pdftex,dvipsnames]{xcolor}

\usepackage[disable]{todonotes}

\usepackage{mathtools}
\usepackage{amssymb}
\usepackage{xspace}
\usepackage{paralist}
\usepackage{dsfont}
\usepackage[ligature,reserved,inference]{semantic}

\usepackage{tikz-cd}

\usepackage{url}

\usepackage{amsmath}
\usepackage{amsfonts}
\usepackage{amssymb}
\usepackage{mathtools}
\usepackage{proof}
\usepackage{colortbl}

\usepackage{hyperref}
\usepackage{nameref}

\usepackage{listings}

\usepackage{wrapfig}

%% file: macros.tex
\mathlig{-x}{\mathbin{-\mspace{-8mu}\times}}
\newcommand\opfail[1]{\overset{{#1}}{-x}}
\mathlig{++}{\mathbin{{+}\mspace{-10mu}{+}}}
\mathlig{->>}{\twoheadrightarrow}

\newcommand{\Om}{\textit{TM}}

\newcommand{\Nat}{\mathds{N}}

\reservestyle{\scheduler}{\textit}
\scheduler{empty,isEmpty,nxt[next],insertSch,pop,bfs,dfs,mon}

\reservestyle{\txexec}{\texttt}
\txexec{nogas[no-gas],txsuc[tx-suc],txfail[tx-fail]}

\reservestyle{\operations}{\mathsf}
\operations{dest,param,money,src}

\newcommand{\textbfup}[1]{\textup{\texttt{#1}}}
\reservestyle{\mech}{\textbfup}
\mech{first,count,fail,setFail,isSelfOnly,txmem[trmem],bstore,ustore,queue}

\reservestyle{\types}{\textit}
\types{Addr,SmartContract}

\reservestyle{\monitors}{\mathit}
\monitors{init,term,begin,end,hookup,fst}

\reservestyle{\infrule}{\textsc}
\infrule{einit,emon,eterm,efail,failinit,failterm,failop,failtx,
   opok[op-ok],opfail[op-fail],bfail[op-begin],efail[op-end]}

\newcommand{\opmem}{\<txmem>\xspace}

\newcommand{\first}{\<first>\xspace}

\newcommand{\countmech}{\<count>\xspace}

\newcommand{\boundedstoragehookup}{bounded storage hookup\xspace}

\newcommand{\QueueInfo}{queue info\xspace}

\makeatletter
\let\orgdescriptionlabel\descriptionlabel
\renewcommand*{\descriptionlabel}[1]{%
  \let\orglabel\label
  \let\label\@gobble
  \phantomsection
  \edef\@currentlabel{#1}%
  \let\label\orglabel
  \orgdescriptionlabel{#1}%
}
\makeatother

\newcommand{\Abegin}{\ensuremath{A_{\<begin>}}}

\newcommand{\Aend}{\ensuremath{A_{\<end>}}}
\newcommand{\Ainit}{\ensuremath{A_{\<init>}}}
\newcommand{\Aterm}{\ensuremath{A_{\<term>}}\xspace}

\newcommand{\partTo}{\rightharpoonup}

%% file: solidity-highlighting.tex
\usepackage{listings}
\usepackage[dvipsnames]{xcolor}

\definecolor{verylightgray}{rgb}{.97,.97,.97}

\lstdefinelanguage{Solidity}{
	keywords=[1]{anonymous, assembly, assert, balance, break, call, callcode, case, catch, class, constant, continue, constructor, contract, debugger, default, delegatecall, delete, do, else, emit, event, experimental, export, external, false, finally, for, function, gas, if, implements, import, in, indexed, instanceof, interface, internal, is, length, library, log0, log1, log2, log3, log4, memory, modifier, new, payable, pragma, private, protected, public, monitor_storage, pure, push, require, return, returns, revert, selfdestruct, send, solidity, storage, struct, suicide, super, switch, then, this, throw, transfer, true, try, typeof, using, value, view, while, with, addmod, ecrecover, keccak256, mulmod, ripemd160, sha256, sha3}, 
	keywordstyle=[1]\color{blue}\bfseries,
	keywords=[2]{set,address, bool, byte, bytes, bytes1, bytes2, bytes3, bytes4, bytes5, bytes6, bytes7, bytes8, bytes9, bytes10, bytes11, bytes12, bytes13, bytes14, bytes15, bytes16, bytes17, bytes18, bytes19, bytes20, bytes21, bytes22, bytes23, bytes24, bytes25, bytes26, bytes27, bytes28, bytes29, bytes30, bytes31, bytes32, enum, int, int8, int16, int24, int32, int40, int48, int56, int64, int72, int80, int88, int96, int104, int112, int120, int128, int136, int144, int152, int160, int168, int176, int184, int192, int200, int208, int216, int224, int232, int240, int248, int256, mapping, string, elem, uint, uint8, uint16, uint24, uint32, uint40, uint48, uint56, uint64, uint72, uint80, uint88, uint96, uint104, uint112, uint120, uint128, uint136, uint144, uint152, uint160, uint168, uint176, uint184, uint192, uint200, uint208, uint216, uint224, uint232, uint240, uint248, uint256, var, void, ether, finney, szabo, wei, days, hours, minutes, seconds, weeks, years},	
	keywordstyle=[2]\color{teal}\bfseries,
	keywords=[3]{block, blockhash, coinbase, difficulty, gaslimit, number, timestamp, msg, gas, sender, sig, value, now, tx, gasprice, origin, add, epochinc, get, setminus, emptyset, system, first, fail, queue, ustore, isSelfOnly, gen_op, verdict, monitor},	
	keywordstyle=[3]\color{violet}\bfseries,
	keywords=[4]{init, term, begin,end},
	keywordstyle=[4]\color{BrickRed}\bfseries,
	identifierstyle=\color{black},
	sensitive=false,
	comment=[l]{//},
	morecomment=[s]{/*}{*/},
	commentstyle=\color{gray}\ttfamily,
	stringstyle=\color{red}\ttfamily,
	morestring=[b]',
	morestring=[b]"
}

%% file: Abstract.tex
\begin{abstract}
  Blockchains are modern distributed systems that provide decentralized
  financial capabilities with trustable guarantees.
  Smart contracts are programs written in specialized programming languages
  running on a blockchain and govern how tokens and cryptocurrency are sent and
  received.
  Smart contracts can invoke other contracts during the execution of
  transactions initiated by external users.
  
  Once deployed, smart contracts cannot be modified and their pitfalls
  can cause malfunctions and losses, for example by attacks from
  malicious users.
  Runtime verification is a very appealing technique to improve the
  reliability of smart contracts.
  One approach consists of specifying undesired executions (\emph{never
  claims}) and detecting violations of the specification on the fly.
  %
  This can be done by extending smart contracts with additional
  instructions corresponding to monitor specified properties, resulting in an
  \emph{onchain} monitoring approach.

  In this paper, we study \emph{transaction monitoring} that consists
  of detecting violations of complete transaction executions and not of
  individual operations within transactions.
  Our main contributions are to show that transaction monitoring is not possible
in most blockchains and propose different execution mechanisms that would
enable transaction monitoring.
\end{abstract}


%% file: Introduction.tex
\section{Introduction}\label{sec:introduction}

%
%
Distributed ledgers (also known as \emph{blockchains}) were first
proposed by Nakamoto in 2009~\cite{nakamoto06bitcoin} in the
implementation of Bitcoin, as a method to eliminate trustable third
parties in electronic payment systems.
Modern blockchains incorporate smart
contracts~\cite{szabo96smart,wood2014ethereum}, which are state-full
programs stored in the blockchain that describe the functionality of
blockchain transactions, including the exchange of cryptocurrency.
Smart contracts allow us to describe sophisticated functionality enabling many
applications in decentralized finances (DeFi), decentralized governance, Web3,
etc.

%
%
Smart contracts are written in high-level programming languages for
smart contracts, like Solidity~\cite{solidity} and
Ligo~\cite{alfour20ligo} which are then typically compiled into
low-level bytecode languages like EVM~\cite{wood2014ethereum} or
Michelson~\cite{michelson}.
Even though smart contracts are typically small compared to
conventional software, writing smart contracts has been proven to be
notoriously difficult.
Apart from conventional software runtime errors (like underflow and overflow),
smart contracts also suffer from new attack patterns~\cite{Daian.2016.DAO} or
from attacks towards the blockchain infrastructure
itself~\cite{Robinson.2020.BlackForest}.
Smart contracts store and transfer money, and are openly exposed to
external users directly and through caller smart contracts.
Once installed the code of the contract is immutable and the effect of
running a contract cannot be reverted (the contract \emph{is} the
law).
%

%
%
There are two classic approaches to achieve software reliability, and there are
attempts to apply them to smart contracts:
\begin{compactitem}
\item \textbf{static techniques} using automatic techniques
  like static analysis~\cite{Stephens.2021.Smartpulse} or model
  checking~\cite{Permenev.2020.VerX}, or deductive software
  verification
  techniques~\cite{ahrendt20functional,nehai19deductive,bhargavan16formal,conchon19verifying},
  theorem
  proving~\cite{Bruno.2019.MiChoCoq,annenkov20concert,schiffl20formal}
  or assisted formal construction of
  programs~\cite{Sergey.2018.Scilla}.
\item \textbf{dynamic
    verification}\cite{ellul18runtime,azzopardi18monitoring,li20securing}
  attempting to dynamically inspect the execution of a contract
  against a correctness specification.
\end{compactitem}
%
%
In this paper, we follow a dynamic monitoring technique.
Monitors are a defensive mechanism where developers write properties
that must hold during the execution of the smart contracts.
If a monitored property fails the whole transaction is aborted.
Otherwise, the execution finishes normally as stipulated by the code
of the contract.

Most of the monitoring techniques inject the monitor into the smart contract as
additional
instructions~\cite{ellul18runtime,azzopardi18monitoring,li20securing}, which is
called inline monitoring~\cite{leucker11teaching}.
The property to be monitored for a method of a given contract $A$ is
typically described as two parts: $\Abegin$, that runs at the
beginning of each call, and $\Aend$, which is checked at the end.
This monitoring code can inspect the storage of contract $A$ and read
and modify specific monitor variables.
For example, monitors can compare the balance at the beginning and
end of the invocation.
However, monitors can only see the contents of $A$ and cannot
inspect or invoke other contracts.
We call these monitors \emph{operation monitors} as they allow us to inspect a
single operation invocation.
In this paper, we study a richer notion of monitoring that can inspect
information across the running transaction, illustrated by our running example.

\paragraph*{Running example: Flash Loans}
%
The aim of a flash loan contract is to allow other contracts to
borrow balance \emph{without any collateral}, provided that the
borrowed money is repaid in the same transaction (perhaps with some
interest)~\cite{EIP.FlashLoan}.
A simple way to specify the correctness of a flash loan contract $A$ is by the
following two informal properties:

\newcommand{\FLSafety}{\textbf{FL-safety}\xspace}
\newcommand{\FLProgress}{\textbf{FL-progress}\xspace}

\vspace{0.2em}
\noindent\begin{tabular}{|l|p{0.8\textwidth}|}\hline
           \FLSafety & {\it No transaction can decrease the balance of \(A\)} \\ \hline
           \FLProgress & {\it A request must be granted unless \textup{\FLSafety} is violated}\\\hline
\end{tabular}
\vspace{0.2em}

Fig.~\ref{fig:listings}(a) shows a simple smart contract attempting to
implement a flash loan lender.
Function \lstinline|lend| checks that the lender contract has enough
tokens to provide the requested loan, saves the initial balance to
check that the loan has been repaid completely, and transfers
the amount requested to the borrower.
Upon return, \lstinline|lend| checks that the loan has been paid back.

\input{Contracts/Three}

Unfortunately, the lender smart contract in Fig.~\ref{fig:listings}(a)
does not fulfill property \FLProgress.
Consider a client, for example Fig.~\ref{fig:listings}(b), that borrows money
from different lenders, then invests the borrowed money to obtain a profit and
finally pays back to the lenders.
In other words, the contract \lstinline|Client| in Fig.~\ref{fig:listings}(b)
collects all the money upfront \emph{before} investing it and then pays back the
lenders.
The contract \lstinline|Client| will not successfully borrow from the lender in
Fig.~\ref{fig:listings}(a), be-cause contract \lstinline|Lender| expects to be
paid back within the scope of method \lstinline|lend|.
However, the contract \lstinline|Client| exercises correctly \FLSafety and
\FLProgress, and returns the borrowed tokens before the transaction finishes.
The problem is that contract \lstinline|Lender| is too \emph{defensive} and only
allows repayments within the control flow of function \lstinline|lend| and
\emph{not in arbitrary points within the enclosing transaction}.
Alternatively, a lender contract could lend funds with the hope that
the client returns the loan before the end of the transaction, but then a
malicious contract, like in Fig.~\ref{fig:listings}(c), would violate
\FLSafety easily.
We cannot solve this problem with operation monitors because it is not possible
within the scope of \lstinline|lend| to successfully predict or guarantee
whether the loan will be repaid within the transaction.

In this article, we propose to extend monitors with two additional
functions: $\Ainit$, which executes before the first call to $A$ in a
given transaction; and $\Aterm$, which executes after the last call to
$A$ (equivalently, at the end of the transaction).
%
As for $\Abegin$ and $\Aend$, $\Ainit$ and $\Aterm$ have access to the
storage and can fail but cannot be called from other contracts or emit
operations.
We call these monitors \emph{transaction monitors} since they can
check properties of the whole transaction.
With transaction monitors, we implement a lender contract that satisfies
\FLSafety and \FLProgress by saving the balance at the beginning of a
transaction in \lstinline|init| and comparing it with the final balance in
\lstinline|term| as shown in Fig.~\ref{fig:trmonitor}.

\input{Contracts/Flashloan_monitor}





As for future work, we envision even more sophisticated monitors that guarantee
properties that involve two or more contracts---like checking that the combined
balance of $A$ and $B$ does not decrease---or even
that predicating about all
\begin{wrapfigure}[7]{l}{0.48\linewidth}
  \vspace{-1.6em}
  \begin{tabular}{|l|l|}\hline
    Global monitors &  future work \\ \hline
    Multicontract monitors & future work \\ \hline
    \rowcolor{gray!35}
    Transaction monitors & this paper \\\hline
    Operation Monitors & \cite{azzopardi18monitoring,ellul18runtime,li20securing} \\\hline
  \end{tabular}
  \caption{Monitors hierarchy}
  \label{fig:hierarchy}
\end{wrapfigure}
contracts participating in a transaction of the whole blockchain.
We refer to them as \textit{multicontract monitors} and \textit{global
  monitors}, respectively, but they are out of the scope of this paper, where we
focus on \textit{transaction monitors}.
Fig.~\ref{fig:hierarchy} shows the monitoring hierarchy.


%
In summary, the contributions of the paper are the following:
\begin{compactitem}
\item The notion of transaction monitors and its formal definition.
\item A proof that current blockchains cannot implement transaction
  monitors, and a list of simple mechanisms that allow their
  implementation.
\item An exhaustive study of how the proposed mechanisms interact with
  each other and the basic building blocks to implement full-fledged
  transaction monitors.
\end{compactitem}

The rest of the paper is organized as follows. 
Section~\ref{sec:modelOfComputation} describes the model of computation. 
Section~\ref{sec:transactionMonitors} studies transaction monitors.
Section~\ref{sec:exec:mechanisms} introduces new execution mechanisms,
and in Section~\ref{sec:results}, we study how these new mechanisms
implement transaction monitors.
Finally, Section~\ref{sec:discussion} concludes.

\vfill\pagebreak


%% file: Contracts/Three.tex
\begin{figure}[t!]
  \begin{tabular}{l}
\begin{lstlisting}[language=Solidity,numbers=none]
contract Lender {
    function lend(address payable dest, uint amount) public {
      require(amount <= this.balance);
      uint initial_balance = this.balance; 
      dest.transfer(amount); 
      assert(this.balance >= initial_balance); 
    }
  }
\end{lstlisting} \\
\multicolumn{1}{c}{(a) A flash loan implementation attempt} \\[0.5em]
\begin{lstlisting}[language=Solidity,numbers=none]
  contract Client {
    Lender l1, l2;
    function borrowAndInvest() public {
      l1.lend(100);l2.lend(200);
      invest(300); 
      l1.transfer(100);l2.transfer(200);
    }
  }
\end{lstlisting}  \\
\multicolumn{1}{c}{(b) A flash loan client} \\[0.5em]
\begin{lstlisting}[language=Solidity,numbers=none]
  contract MaliciousClient {
    Lender l;
    function borrowAndInvest() public {
      l.lend(100); 
      invest(100);
    }
  }
\end{lstlisting} \\
\multicolumn{1}{c}{(c) A malicious flash loan client}
  \end{tabular}
  \label{fig:listings}
  \caption{Pseudocode for contracts \lstinline|Lender|, \lstinline|Client| and \lstinline|MaliciousClient|.}
\end{figure}

%% file: Contracts/Flashloan_monitor.tex
\begin{figure}[t!]
\begin{lstlisting}[language=Solidity,numbers=none]
  contract Lender {
    function lend(address payable dest, uint amount) public {
      require(amount <= this.balance);
      dest.transfer(amount);
    }  
  } with monitor {
    uint initial_balance;    
    init { initial_balance = this.balance; }
    term { assert(this.balance >= initial_balance); }
  }
\end{lstlisting}
\caption{A correct flash loan implementation using transaction monitors}
\label{fig:trmonitor}
\end{figure}


%% file: ModelOfComputation.tex
\section{Model of Computation}\label{sec:modelOfComputation}

We introduce now a general model of computation that captures the
evolution of smart contract blockchains.

\subsubsection{An Informal Introduction.}
Blockchains are a public incremental record of the executed
transactions.
Even though several transactions are packed in ``blocks''--- which are
totally ordered---, transactions within a block are also totally
ordered.
Therefore, we can interpret blockchains as totally ordered sequences of
transactions.

Transactions are in turn composed of a sequence of operations
where the initial operation is an invocation from an external user.
Each operation invokes a destination contract (where contracts are identified by
their unique address).
Operations also contain the name of the invoked method,
arguments and balance (in the cryptocurrency of the underlying
blockchain), and an amount of gas\footnote{The notion of gas is
  introduced to make all operations terminate because each
  individual instruction consumes gas and once the initial operation
  is invoked no more gas can be added to the transaction.}.
The execution of an operation follows the instructions of the program
(the smart contract) stored in the destination address.

Given the arguments and state of the blockchain, the code of every
smart contract is \emph{deterministic} which makes the blockchain
predictable and amenable to validation.
We model smart contracts as pure computable functions taking their input
arguments and the current local storage of the contract,
and returning (1) the changes to be performed in the local storage; (2) a
list of further operations to be executed.
No effect takes place in their local storage until the end of the
operation.
This abstraction does not impose any restriction since every imperative program
can be split into a collection of basic pure code blocks separated by the
instructions with effects.

The execution of a transaction consists of iteratively executing
pending operations, computing their effects (including updating
the pending operations) until either (1) the queue of pending
operations is empty, or (2) some operation fails or the gas is
exhausted.
In the former case, the transaction commits and all changes are made
permanent.
In the latter case, the transaction aborts and no effect takes place
(except that some gas is consumed).

\subsubsection{Model of Computation.}
We now formally model the state of a blockchain during the execution
of the operations forming a transaction.
We represent a blockchain configuration as a pair $(\Sigma,\Delta)$
where:
\begin{description}
\item[Blockchain state] $\Sigma$ is a partial map between addresses
  and the storage and balance of smart contracts,
\item[Blockchain context] $\Delta$ contains additional information about the
blockchain, such as block number, current time, amount of money sent in the
  transaction, etc.
\end{description}
Blockchain contexts may vary since different blockchains carry different
information, but either implicitly or explicitly, every
blockchain maintains a blockchain state.
The computation of a successful transaction begins with an external
operation \(o\) from a configuration $(\Sigma,\Delta)$ and either
aborts or finishes into a final configuration $(\Sigma',\Delta')$.
%
%
%

We model a \emph{smart contract} as a partial map
$A : \Delta \times \bbbp \times \bbbs \times \bbbn \rightharpoonup
(\bbbs \times [\mathcal{O}])$
where \(\bbbp\) is the set of all possible parameters of \(A\),
\(\bbbs\) the set of all possible storage states, \(\mathcal{O}\) the
set of operations and \([\cdot]\) is a set operator representing lists
of elements of a given set.
Smart contracts written in imperative languages with effects can be
modeled as sequences of pure blocks where effects happen at the end in
the standard way.

 %
%

\paragraph{Operations.}
An operation is a record containing the following fields:
\begin{compactitem}
\item \(\<dest>\) the address to invoke;
\item \(\<src>\) the address initiating the operation;
\item \(\<param>\) parameters expected by the smart contract at address \(\<dest>\);
\item \(\<money>\) the amount of crypto-currency sent in the operation.
\end{compactitem}
We use standard object notation to access each field, so
\(o.\<dest>\) is the destination address, \(o.\<src>\) is the source
address, \(o.\<param>\) the parameters and \(o.\<money>\) the
amount transferred.
%


\paragraph{Transactions.}
A transaction results from the execution of a sequence of operations
starting from an external operation placed by an external user.
If an operation fails the transaction fails and the blockchain
state remains unchanged.
A successful operation \(o\) results in a new storage and a list of new operations
\(ls\).
The blockchain updates the storage of smart contract \(o.\<dest>\) and balance 
of both smart contracts \(o.\<dest>\) and \(o.\<src>\)
generating a new blockchain configuration and the list \(ls\) is added to the
current pending queue of operations.
Operations are executed one at a time modifying the blockchain
configuration until some operation fails or there is no more
operations on the pending queue.
In the second case, the transaction is successful and the last blockchain
configuration consolidates.

We assume there is an implicit partial map from addresses to smart contracts
$\mathbb{G} : \<Addr> \partTo \<SmartContract>$.
Moreover, we assume map \(\mathbb{G}\) does not change since we assume that smart
contracts cannot install new contracts.

\paragraph{Operation Execution.}
Let $o$ be an operation and $(\Sigma,\Delta)$ a blockchain configuration.
The evaluation of $o$ from $(\Sigma,\Delta)$ results in a new
configuration and a list of operations \(ls\), 
which we denote $(\Sigma,\Delta) \xrightarrow{o} (\Sigma',\Delta',ls)$ whenever:
\begin{compactenum}
\item The source smart contract has enough balance,
  \(\Sigma(o.\<src>) \geq o.\<money>\)
\item The invocation to the smart contract is successful:
  \[ \mathbb{G}(o.\<dest>)(\Delta,o.\<param>,\Sigma(o.dst).st,\Sigma(o.dst).balance) = (st',ls)\]
\end{compactenum}
The new blockchain configuration state $\Sigma'$ is the result of: 1) adding
\(o.\<money>\) into the balance of \(o.\<dest>\) and subtracting it
from \(o.\<src>\), and 2) updating the storage as
$\Sigma'(o.\<dest>).st=st'$.
Note that we leave the evolution of \(\Delta\) unspecified as it is
system dependant.
In Section~\ref{sec:results}, we implement different additional blockchain
features by inspecting (and possibly modifying) the blockchain context.
For failing evaluation of operations, we use
$(\Sigma,\Delta) \opfail{o}$.
%

\paragraph{Execution Order.}
%
The execution can proceed in different ways.
We consider two execution orders: new operations are added to the
beginning of the pending queue (a DFS strategy) and new operations added to
the end of the pending queue (a BFS strategy).
This results in the following transition rules:
\[
 \begin{array}{c}
    \inference*
      {(\Sigma,\Delta) \opfail{o}}
      {(\Sigma,\Delta, o :: os) \not\leadsto_{a}} \\
  \begin{tabular}{cc}
    \inference*
      {(\Sigma,\Delta) \xrightarrow{o} (\Sigma', \Delta', ls)}
      {(\Sigma,\Delta, o :: os) \leadsto_{\<dfs>} (\Sigma', \Delta', ls ++ os)}
      &
    \inference*
      {(\Sigma,\Delta) \xrightarrow{o} (\Sigma', \Delta', ls)}
      {(\Sigma,\Delta, o :: os) \leadsto_{\<bfs>} (\Sigma', \Delta', os ++ ls)}
  \end{tabular}
  \end{array}
\]
The execution starting from an external operation \(o\) is a sequence
of steps \((\leadsto_{a})\)---with \(a\) fixed to be either \<dfs> or
\<bfs>---until the pending operation list is empty or the execution of
the next operation fails.
Beginning from a blockchain configuration $(\Sigma,\Delta)$ and an
initial operation $o$, a transaction execution is a sequence of
operation executions: \(
  (\Sigma, \Delta, [o]) \leadsto_{a} (\Sigma_{1}, \Delta_{1}, os_{1})
  \leadsto_{a} \ldots \leadsto_{a} (\Sigma_{n}, \Delta_{n}, []) 
\)
or that
\(
  (\Sigma, \Delta, [o]) \leadsto_{a} (\Sigma_{1}, \Delta_{1}, os_{1}) 
  \leadsto_{a} \ldots \leadsto_{a} (\Sigma_{n}, \Delta_{n}, os_{n})\not\leadsto_a
\)
%

A transaction can fail either because of gas exhaustion or an internal operation
has failed, and in that case, we have a sequence of \(\leadsto_{a}\) leading to a final
step marked as \(\not\leadsto_a\) following the failing operation.

Finally, after every successful execution, the blockchain takes the
last configuration and upgrades its global system.

The model of computation described in this section does not follow
exactly a call-and-return model like the Ethereum blockchain
does~\cite{wood2014ethereum}.
However, it is easy to see that it can be simulated in our model by
having each contract explicitly keeping its stack of returned values.

%% file: TransactionalMonitors.tex
\section{Transaction Monitors and Execution Mechanisms}
\label{sec:transactionMonitors}

We now introduce \emph{transaction monitors} and show that it is not
possible to implement them in current blockchains.
We present different extensions that allow us to implement
transaction monitors.

\subsection{Transaction Monitors}\label{sec:Transaction:Monitors}
Transaction monitors allow us to reason about properties of transactions.
Each smart contract $A$ is equipped with a monitor storage and four
especial methods $A_{\<init>}$,  $A_{\<begin>}$, $A_{\<end>}$ and ${\Aterm}$.
These new methods cannot emit operations or modify smart contract
storage, however, they have their own monitor storage.
We assume that these new methods are interpreted by the blockchain and if one of
these methods fail the whole transaction fails.
Otherwise, the effect in the blockchain is the same
as if it was executed without monitors.
The functions $A_{\<init>}$ and $\Aterm$ can read the storage and
balance of the smart contract and read and write the monitor storage.
Function \(A_{\<init>}\) is executed before the first time $A$ is
invoked in the transaction and function \(\Aterm\) is invoked
after the last interaction to $A$ finished in the transaction, and
does not modify the monitor storage.
Functions $A_{\<begin>}$ and $A_{\<end>}$ are executed at the beginning and
at the end of each \emph{operation} that is executed in $A$, as
in operation monitors~\cite{ellul18runtime} (note that $A_{\<begin>}$
and $A_{\<end>}$ can be easily implemented by inlining their code
around the methods of $A$).
The method $A_{\<begin>}$ takes the same arguments as any $A$
operation plus the monitor storage, while function $A_{\<end>}$ has
access to the result of the operation (list of the operation emitted
and the new storage) plus the monitor storage.
We call the resulting smart contracts \emph{monitored smart
  contracts}.

%

%

\paragraph{Operation Monitors.}
We first extend the model of computation to include operation monitors.
A monitored operation execution is a normal operation execution where the
corresponding operation monitor is executed before and after the operation is
executed.

We define $(\xrightarrow[\<mon>]{o})$ modifying $(\xrightarrow{o})$ as
follows.
Before executing $o$, (1) procedure $\mathbb{G}(o.\<dest>).\<begin>$ is
invoked, then (2) operation $o$ is executed, and (3) finally
$\mathbb{G}(o.\<dest>).\<end>$ runs.
That is, operation monitors are simply restricted functions executed
before and after each operation.
We can then specialize \(\leadsto_{a}\) with operation monitors,
that is, use relation \((\xrightarrow[\<mon>]{o})\) instead of
relation \((\xrightarrow{o})\) to obtain transaction executions that
use operation monitors.

Procedures \(\<begin>\) and \(\<end>\) can only modify the
private monitor storage and fail, and thus, they cannot interfere
in the normal execution of smart contracts (except by failing more
often).

\paragraph{Transaction Monitors.}

We redefine transaction monitors execution as a restriction of the
transaction execution relation so transactions invoke \(\<init>\)
and \(\<term>\) when required.
In this case, \(\<init>\) can change the monitor storage, and thus,
can modify the blockchain state.
We define a new relation \({->>}_{a}\) the smallest relation defined
by the following inference rules:

\[ \begin{array}{ccc}
      \inference
      {A_{\<init>}(\Sigma(A)) = \Sigma'}
      {(\Sigma,\Delta,os) {->>}_{a} (\Sigma', \Delta,os)} &
      \inference                                                            
      {\Aterm(\Sigma(A))}
      {(\Sigma,\Delta,[]) {->>}_{a}(\Sigma,\Delta, [])} &
      \inference
      {(\Sigma,\Delta,o :: os) {\leadsto}_{a} (\Sigma',\Delta',os')}
      {(\Sigma,\Delta,o :: os) {->>}_{a} (\Sigma',\Delta',os')}
   \end{array}
\]

%
Note that we sacrifice a deterministic operational semantics in favor of
a clearer set of rules.
As before, we use \((-x)\) to represent failing transactions.

\[
  \begin{array}{ccc}
        \inference
        {(\Sigma,\Delta) \opfail{o}}
        {(\Sigma,\Delta,os) {-x}} &
        \inference
        {A_{\<init>}(\Sigma(A)) \opfail{}}
        {(\Sigma,\Delta,os) {-x}} &
        \inference
        {A_{term}(\Sigma(A)) \opfail{}}
        {(\Sigma,\Delta,[]) {-x}}
  \end{array}
\]
%
%
Finally, we define a monitored trace of a transaction same as before,
given a blockchain configuration \((\Sigma,\Delta)\) and an external
operation $o$:
\[
  (\Sigma, \Delta, [o]) {->>}_{a} (\Sigma_{1}, \Delta_{1}, os_{1}) {->>}_{a} (\Sigma_{2},\Delta_{2}, os_{2}) {->>}_{a} \ldots
  {->>}_{a} (\Sigma_{n},\Delta_{n},[])
\]

To remove the non-determinism we add a new relation that restricts the
legal runs.
This relation knows the set of visited addresses (smart contracts),
and invokes an initialization method, and at the very end of the
evaluation of a transaction uses the same set to invoke their
corresponding term method.
\[
  \begin{array}{c}
  \inference
  { (\Sigma,\Delta, os) {->>}^{o}_{a} (\Sigma', \Delta', os')
    & \hspace{3em}o.\<dest> \in E
  }
  { E \vdash (\Sigma,\Delta,o :: os) \Rightarrow_{a} E \vdash (\Sigma', \Delta', os')} \\
    \vspace{0.8em} \\
  \inference
    { (\Sigma'',\Delta,os) \leadsto^{o}_{a} (\Sigma', \Delta', os')
    & \hspace{2em}o.\<dest> \notin E \hspace{2em}
    (\Sigma,\Delta, os) {->>}^{A_{\<init>}}_{a} (\Sigma'', \Delta, os)
  }
  { E \vdash (\Sigma,\Delta,o :: os) \Rightarrow_{a} E \cup \{o.\<dest>\} \vdash (\Sigma', \Delta', os')} \\
    \vspace{0.8em} \\
  \inference
    {
    (\Sigma,\Delta, []) {->>}^{A_{\<term>}}_{a} (\Sigma, \Delta, os)
    & \hspace{2em}e \in E
    }
    {E \vdash (\Sigma,\Delta, []) \Rightarrow_{a} E \setminus \{e\} \vdash (\Sigma, \Delta, []) }
  \end{array}
\]

As result, we only accept traces generated by relation
\((\Rightarrow_{a})\), beginning with a blockchain configuration
\((\Sigma,\Delta)\) and an external operation \(o\) resulting in
failure or a new blockchain configuration \((\Sigma', \Delta')\):
\( \emptyset \vdash (\Sigma,\Delta,[o]) \Rightarrow_{a} \ldots \Rightarrow_{a} \emptyset \vdash (\Sigma', \Delta', []).\)

\subsection{Transaction Monitors in BFS/DFS}\label{sec:impossibility}

Unfortunately, transaction monitors cannot be implemented in
blockchains that follow DFS or BFS evaluation strategies.
We show now a counter-example.
Consider a transaction monitor for $A$ that fails when smart contract
$A$ is called \textbf{exactly} once in a transaction.
The monitor storage contains a natural number which function
\(\<init>\) sets to $0$, $\<begin>$ adds one to the counter, $\<end>$
does nothing, and $\<term>$ fails if the monitor storage is exactly
one.

Now let $(\Sigma,\Delta)$ be a blockchain configuration, and let $A$
and $B$ be two smart contracts, where $A$ is being monitored for the
``only once'' property.
Consider the following two executions of external operations from
$(\Sigma,\Delta)$:
\begin{compactitem}
\item $o_1$ invokes $B.f$ which then invokes $o_{A1}$ in $A$,
\item $o_2$ invokes $B.g$ which then invokes $o_{A1}$ and $o_{A2}$ in $A$.
\end{compactitem}
The monitor for ``only once'' must reject the transaction beginning
with \(o_{1}\), but accept the transaction beginning with \(o_{2}\).

Consider a DFS strategy.
Starting from \(o_{1}\), the execution trace is
\[
  (\Sigma,\Delta,[o_{1}]) \leadsto_{\<dfs>} (\Sigma_{1},\Delta_{2},[o_{A1}]) \leadsto_{\<dfs>} (\Sigma_{2},\Delta_{2},as_{1})
\]
with corresponding sequence of pending operations $[o_1]$, $[o_{A1}]$, $as_1$.
Starting from $o_{2}$ the sequence of pending operations is $[o_2]$,
$[o_{A1};o_{A2}]$, $as_1++[o_{A2}]$,\ldots,$[o_{A2}]$, $as_2$.
%
%
It is not possible to distinguish between the traces generated by
\(o_{1}\) and \(o_{2}\), as anything that operation $o_{A1}$ and its
descendants \(as_{1}\) do will happen before the execution of $o_{A2}$
in the second transaction.
In other words, \(o_{1}\) and all the operations that can be generated
by it or its descendants 
cannot know that some other invocation to \(A\) is pending so $A$
cannot fail preventively.
%
%
At the same time $o_{A1}$ is the only chance in $A$ to make the first
transaction fails because there is no other operation in $A$.
Therefore, the two runs are identical up to the end of $o_{A1}$ and
one must fail and the other must not fail.

A BFS scheduler can distinguish between the execution of operations
\(o_{1}\) and \(o_{2}\) by using a \emph{recurring operation}.
Since new operations are added to the end of the pending queue, $A$
can inject an operation that checks $A$'s state and conditionally (if
the test that would make $\<term>$ fail is true) injects itself at the
end of the pending queue.
If the condition that makes $\<term>$ accept is never met, the
transaction fails because the recurring operation injects itself
ad-infinitum, exhausting gas.
In Section~\ref{s:results:BFSScheduler} we use recurring
operations thoroughly.
However, a simple variation of this example that includes comparing
with a third transaction 
where $A$ is invoked three times shows that 
BFS cannot implement ``only once'' either (as BFS cannot distinguish
between the third invocation to $A$ and a first invocation to $A$ in 
a transaction following the one originated by $o_2$).
%
For a detailed proof see Appendix~\ref{app:bfs:imp}.

%% file: ExecutionMechanisms.tex
\section{Execution Mechanisms}\label{sec:exec:mechanisms}

We propose new mechanisms and study if they help to implement transaction
monitors.
However, adding features to blockchains is potentially dangerous since
it can introduce unwanted behaviour\cite{Daian.2016.DAO}.
We focus on simple mechanisms that are easy to implement and are
backwards compatible.

Since $\Abegin$ and $\Aend$ can already be implemented using inlining,
we focus on mechanisms that allow executions at the beginning and end
of transactions, which can aid to implement $\Ainit$ and $\Aterm$.
We present two kinds of mechanisms, ones that introduce a new
instruction, and others that add a new special method to smart
contracts.
In the next section, we compare their relative power and if they can
implement transaction monitors.

\subsubsection{Mechanisms that Add New Instructions.}
The first four mechanisms add new instructions and can be easily implemented by
bakers/miners collecting the information required in the context \(\Delta\).
\begin{compactitem}
\item \textit{First.}  We consider a new instruction, \<first>, which
  returns true if the current operation is the first invocation
  to the smart contract in the current transaction.
  The context $\Delta$ can be extended to contain the set of contracts
  $F$ that have already run an operation in the current transaction,
  which allows us to implement \<first> as $A\not\in F$, where $A$
  is the smart contract that executes \<first>.
\item \textit{Count.}  We introduce now a new instruction, \<count>
  that returns how many invocations have been performed to methods
  of the contract in the current transaction.
  Again, the context $\Delta$ can easily count how many times each
  contract has been invoked.
\item \textit{Fail/NoFail.}
  This mechanism equips each contract with a new flag \<fail>
  that can be assigned during the execution of the contract (and that
  is false by default).
  The semantics is that at the end of the transaction, the whole
  transaction would fail if some contract has the \<fail> bit to true.
  For example, the failing bit allows us to implement flash loans as
  follows. A lender smart contract can set \<fail> to true when is
  lending money and change it to false only when the money is
  returned.
\item \textit{Queue info.}  We add a new operation, \<queue>,
  indicating if there is no more interaction between smart contracts.
  Or equivalently, if the only operations permitted in the pending queue are
  recurrent operations (which can only inject operations to the same contract).
  These operations must also be specially qualified in the contract,
  and the runtime system must make sure that they only generate
  operations to the same contract.
\end{compactitem}

\subsubsection{Mechanisms that Add New Methods or Storage.}

The following mechanisms modify the definition of smart contracts
either by adding new methods that are executed at particular moments
in a transaction or by adding special storage/memory.
\begin{compactitem}
\item \textit{Transaction Memory.}  Smart contracts are equipped with
  a special volatile memory segment that exists only during the
  execution of a transaction and which is created and initialized at
  the beginning of the transaction.
  We add a new segment in the smart contract indicating the initial values to be
  assigned.
  In concrete, each contract \(A\) indicates a new storage type for
  the transaction memory and a procedure that initializes it (which
  can read but not change the conventional storage).
  We use \<txmem> to refer to this mechanism.
\item \textit{Storage Hookup, Bounded and Unbounded.}
  The idea is to equip smart contracts with a new method that updates
  the storage after the last local operations in the transaction.
  %
  These methods can only modify the storage but
  not invoke other methods.
  A bounded version of this mechanism is restricted to terminating
  non-failing functions (for example, by restricting the class of
  programs).
  In addition, the unbounded version is arbitrary code that can fail.
  We use \<bstore> and \<ustore> to refer to these mechanisms.
\end{compactitem}


%% file: results.tex
\section{Implementing Transaction Monitors}
\label{sec:results}

We say a mechanism $M$ implements another mechanism $N$ if and only if assuming
a blockchain equipped with $M$ can write every smart contract that a blockchain
equipped with mechanism $N$ can write.
We say that two mechanisms are equivalent if and only if they can implement each
other.
We disregard gas consumption here, only considering infinite computations (i.e.
we assume that one can always assume sufficient gas).

\begin{theorem}
  \label{th:first:count:tm:bsh}
  The following are equivalent: \opmem, \first, \countmech, and \<bstore>.
\end{theorem}

If contracts can know when their first invocation in the transaction
occurs, they can set the storage in different ways simulating
\countmech and \opmem.
Also, \countmech and \opmem can simulate \first, by checking if the
count is $0$ and initializing a volatile bit to \(true\).
More interesting is that \first can simulate \boundedstoragehookup
by applying the effect on the storage of
\boundedstoragehookup at the beginning of the next transaction.
Detailed proofs are included in Appendix~\ref{app:equiv}.

\begin{lemma}
  \label{l:ush:bsh}
  \label{l:ush:fnfhookup}
  Mechanism \<ustore> implements \<bstore> and \<fail>.
\end{lemma}

\begin{proof}
  Mechanism \<ustore> implements \<bstore> trivially as it is just less
  restrictive.
  For \<fail> we add in the storage of $A$ a new field, \textit{fl} to
  represent the failing bit which is initialized to false when the
  contract is installed and updated to simulate the \<fail> instruction.
  At the end of the transaction, the \<ustore> hookup checks if
  \textit{fl} is true and fail.  Otherwise, it does nothing.\qed
\end{proof}

It can be proven that the other direction is not always possible.
Fig~\ref{fig:relations} shows graphically the previous results where
an arrow indicates that one mechanism implements another.
In this diagram, an absence of an arrow does not necessarily imply
impossibility but perhaps that the result depends on the execution order.
For example, in BFS blockchains \<first> can implement \<ustore>, but
this is impossible with DFS.

\begin{figure*}[t!]
    \centering
    \includegraphics[scale=0.5]{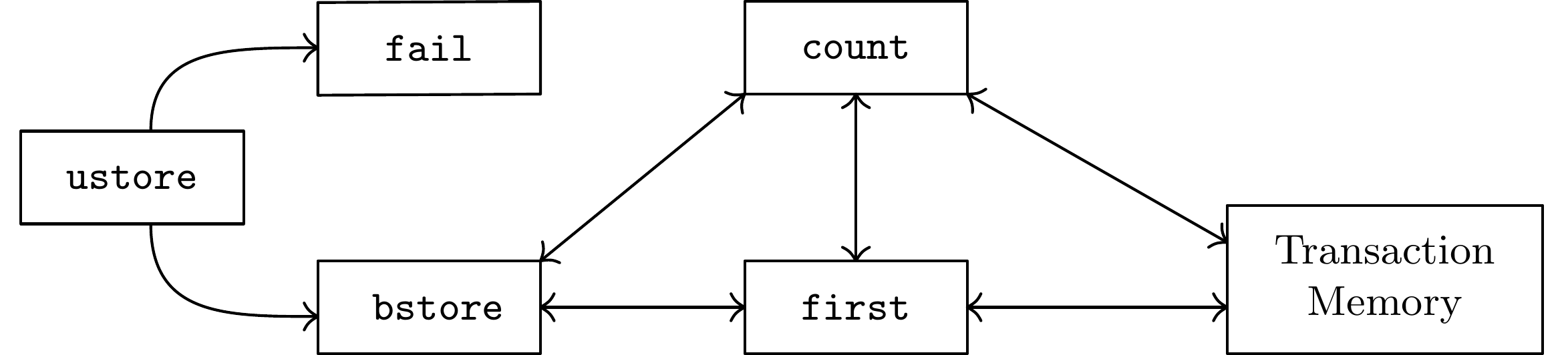}
    \caption{Relation between mechanisms for any scheduler. An arrow from
mechanism M to mechanism N means that M implements N.}
    \label{fig:relations}
\end{figure*}

Since \<first>, \countmech, \<bstore> and \opmem are all equivalent, from
now on we only refer to mechanism \first. 
It is easy to see that this mechanism is enough to implement \(\<init>\).

To implement \(\<term>\), we can either implement \<fail> or
\<ustore>, where \<fail> is simpler, and \<ustore> is more powerful but
requires a bigger change to blockchains.
\begin{theorem}
  \label{thm:FstFnF:TMon}
  Mechanisms \<first> \(+\) \<fail> implement transaction monitors.
\end{theorem}
\begin{proof}
  Let \(\mathcal{B}\) be a blockchain that implements \<first> and
  \<fail>.
  Given a monitored smart contract \(A\), we want to implement \(A\) in
  blockchain \(\mathcal{B}\).
  We define a new smart contract \(A'\) extending its storage to also
  contains \(A\)'s monitor storage.
  Then, we equip $A'$ with a new method $f'$ for every method $f$ in
  $A$, such that, $f'$ first checks \<first> and executes
  \(A_{\<init>}\) if needed. Then, before exiting, $f'$ executes
  \(A_{\<term>}\) with the current state but instead of failing
  explicitly $f'$ set the failing bit.
  Function \(A_{\<init>}\) is executed exactly once and
  \(A_{\<term>}\) may be executed multiple times, but it does not
  modify the contract storage and it does not generate operations.
  The last execution of $\Aterm$ in \(A'\) will simulate $\Aterm$ in $A$.
  If the semantics of the blockchain were such that the balance of
  pending outgoing operation would subtract balance from $A$ when it executes,
  then these calculations can be made in the monitor storage when the operations
are generated.\qed
\end{proof}


Since \<ustore> implements \<first> and \<fail>, it follows that
\<ustore> implements transaction monitors.
\begin{corollary}\label{c:UStorageHookup:Monitors}
  \<ustore> implements transaction monitors but transaction monitors
  cannot implement \<ustore>.
\end{corollary}

Transaction monitors can only make contracts fail but not change the
storage.
Our results are summarized in Fig.~\ref{fig:relations-monitors}.

\begin{figure*}[t!]
    \centering
    \includegraphics[scale=0.5]{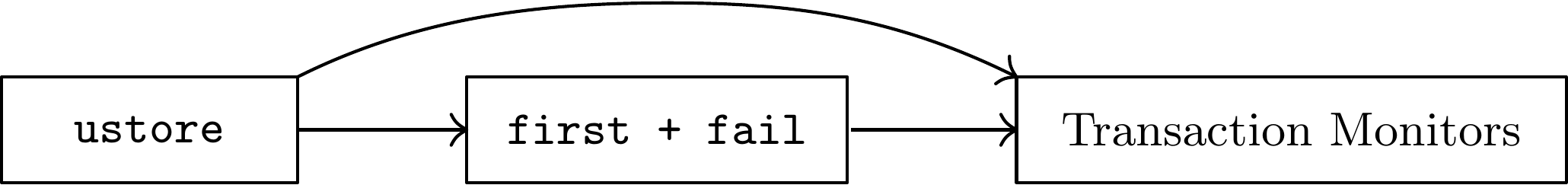}
    \caption{Relation between mechanisms and transaction monitor for any scheduler.}
    \label{fig:relations-monitors}
\end{figure*}

\subsection{BFS Blockchains} \label{s:results:BFSScheduler}

We now study more in detail the mechanism for BFS based blockchains.
The first result is that unless equipped with further mechanisms, BFS
blockchains cannot implement transaction monitors.
The essence of the proof is to create two transactions on a monitored contract
$A$ (like in ``only once'') in which corresponding invocations to the
same contract $A$ receive identical information, and one must fail and
the other commit.

\newcounter{thm-bfs-monitors}
\setcounter{thm-bfs-monitors}{\value{theorem}}

\begin{theorem}
  \label{thm:bfs:monitors}
  A BFS blockchain does not implement transaction monitors.
\end{theorem}

A BFS blockchain guarantees that new operations are executed after all
pending operations, which enables the implementation of implement
\<fail> using recurring operations.
A recurring operation is a private function that can read and write
the storage and that either terminates or reinjects itself again to
the pending queue.
Since every time the operation is executed the blockchain consumes
gas, and eventually, failure follows from an attempt to inject itself
ad-infinitum.

\begin{lemma}\label{lmm:BFS:FailNoFail}
  Recurring operations in BFS blockchain allow to implement \<fail>.
\end{lemma}



Since transaction monitors cannot be implemented within a BFS
blockchain~(see Appendix~\ref{app:bfs:imp}), we conclude that \<fail>
does not implement transaction monitors in BFS blockchains.
The missing element is \<first> which allows to implement \<ustore>.
And, since \<ustore> implements transaction monitors
(Corollary~\ref{c:UStorageHookup:Monitors}), \<first> can
also implement transaction monitors. 

\begin{lemma}\label{l:BFSfirst:ush}
  Mechanism \<first> implements \<ustore> in BFS blockchains.
\end{lemma}
\begin{proof}
  Assume a BFS blockchain implementing \first.
  Let $A$ be a smart contract.
  We modify $A$ to contain a second copy $S'$ of its storage.
  Upon the first call of \(A\), we update the current storage using
  the values in \(S'\).
  We add a new private method $\<hookup>$ in $A$ that mimics the code of
  \<ustore> but (1) it applies the changes in $S'$, and (2) instead of
  failing (if \<ustore> fails) it calls itself as a recurring
  operation.
  Finally, we modify $A$ so that function $\<hookup>$ is invoked at
  the end of each method in $A$.
  In effect, $\<hookup>$ is preventively evaluating \<ustore> on the
  side memory $S'$, and simulating the failure as a recurring
  operation (when \<ustore> fails).
  Therefore, if the operation is the last one on the contract and
  it does not fail, then $S'$ contains the correct storage, which will
  then be copied at the beginning of the next transaction.\qed
\end{proof}
In the previous proof, we split mechanism \<ustore> into two parts:
one in charge of updating the storage, the other in charge of failing.
If we also add \<queue>, we can implement \<ustore> without failing by
gas exhaustion because now the $\<hookup>$ executed recurrently can
know if there are only recurrent operations and then execute the
\<ustore> code (including the failure).

\begin{lemma}\label{l:BFSqueueinfo:ush}
  Mechanism \<queue> implements \<ustore> in BFS
  blockchains.
\end{lemma}

In a BFS blockchain, \<ustore> implements transaction monitors
(Corollary~\ref{c:UStorageHookup:Monitors}), and thus, by the previous
lemma, \<queue> also implements transaction monitors.
Next, we will show that \<queue> cannot be implemented with
\<ustore> when a BFS strategy is used.
Intuitively, mechanism \<queue> adds a way for smart contracts to know
the state of the blockchain, i.e. if there is still interaction
between smart contracts, and thus, smart contracts can take different
actions based on the state of the blockchain, while mechanism
\<ustore> adds a way to execute a procedure at the end of
transactions, but smart contracts are oblivious about interactions
between smart contracts.
Since \<ustore> implements all other mechanisms,
we have that no other mechanism can implement \<queue>.

\newcounter{lemma-bfs-queue}
\setcounter{lemma-bfs-queue}{\value{lemma}}
  
\begin{lemma}
  \label{l:bfs:queue}
  In BFS blockchains \<ustore> cannot implement \<queue>.
\end{lemma}

The main idea is to create two executions that are identical unless one can
inspect the pending operation queue, and in which one operation must fail if
\<queue> returns that the queue of pending operations is empty.
The complete proof is in Appendix~\ref{app:bfs:imp}.
Fig~\ref{fig:relations-bfs} summarizes the relations between
mechanisms and transaction monitor in BFS blockchains.

\begin{figure*}[t!]
  \centering
  \includegraphics[scale=0.5]{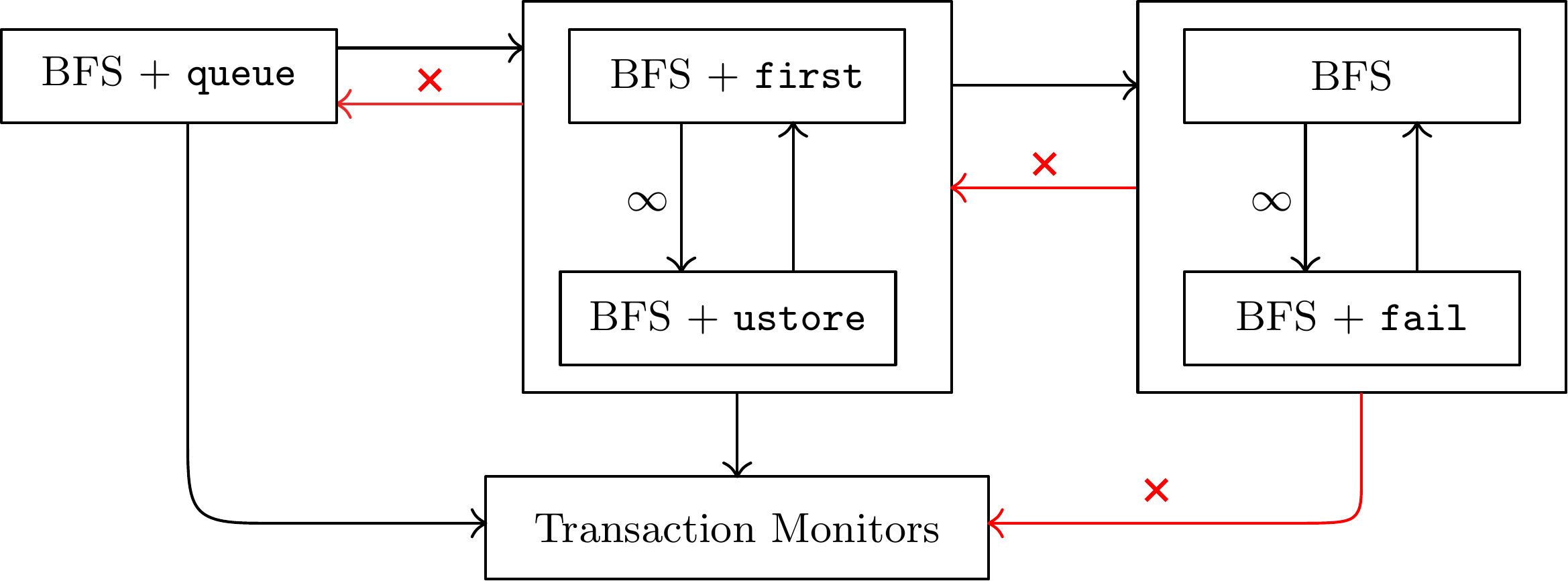}
  \caption{Relation between mechanisms and transaction monitor in BFS blockchains.
A black arrow from mechanism $M$ to mechanism $N$ means that $M$ can
implement $N$.
The \(\infty\) symbol represents the use of an infinite recursion to provoke a failure.
A red arrow with a cross from mechanism $M$ to mechanism $N$ means
that $M$ cannot implement $N$.  }
    \label{fig:relations-bfs}
\end{figure*}

\subsection{DFS Blockchains}\label{s:results:DFSScheduler}

We now study DFS blockchains, that is, when the resulting list of
operations from smart contracts execution are appended at the beginning
of the list.
This is the most conventional execution order in most blockchains, like
Ethereum.
We now prove several impossibility results.


Mechanisms \<ustore> and \<first> plus \<fail> implement transaction
monitors (Corollary~\ref{c:UStorageHookup:Monitors} and
Thoerem~\ref{thm:FstFnF:TMon}).
In a DFS blockchain, those are the only two ways using our mechanisms
to implement transaction monitors.
We show that transaction monitors cannot be implemented by combining
\<queue> with either \<first> or \<fail>, and as a consequence none of
these mechanisms on their own can implement transaction monitors.

\newcounter{lem-dfs-queue-first}
\setcounter{lem-dfs-queue-first}{\value{lemma}}

\begin{lemma}
  \label{lem:dfs-queue-first-nomon}
  A DFS blockchain implementing \<queue> and \<first> does no
  implement transaction monitors.
\end{lemma}

\begin{proof}
  Let \(\mathcal{B}\) be a DFS blockchain and \(A\) a smart contract installed
  in \(\mathcal{B}\).
  %
  Consider the ``only once'' monitor that fails if and only if the smart
  contract $A$ is called exactly once.
  We show that this monitor cannot be implemented in DFS even with
  \first and \<queue>.
  
  Let \(B,C\) be two other smart contracts.
  We analyze the pending queue of execution of two possible external operations
  originated by \(B\):
  \begin{compactenum}
  \item \(o_{1}\) where $B$ calls $A$ \emph{once} and then $C$
  \item \(o_{2}\) where $B$ calls $A$ \emph{twice} and then $C$
  \end{compactenum}
  We assume that there are no additional invocations to \(A\) aside from the
  described above.
  When we execute both operations in a \((\Sigma,\Delta)\) blockchain system,
  we have the following two traces:
  \[
    \begin{array}{rcl}
      \bullet & t_{1} : & (\Sigma,\Delta, [o_{1}]) \leadsto_{\<dfs>} (\Sigma',\Delta', [a_{1},c_{1}]) \ldots \\
      \bullet & t_{2} : & (\Sigma,\Delta, [o_{2}]) \leadsto_{\<dfs>} (\Sigma',\Delta', [a_{1},a_{2},c_{1}]) \ldots
    \end{array}
  \]
Note that the presence of operation \(c_{1}\) in the pending execution queue is
forcing mechanism \<queue> to return false.
Since the occurrence of operation \(a_{1}\) in both cases execute in
the same configuration, the behavior must be the same.
The transaction executing \(o_{1}\) must fail because $A$ is called
only once, but this will make the second transaction fail as well.\qed
\end{proof}
We can conclude that neither \<queue> nor \first alone would implement
transaction monitors.

\newcounter{lem-dfs-fail-tmons}
\setcounter{lem-dfs-fail-tmons}{\value{lemma}}

\begin{lemma}
  \label{lmm:DFS:QInfoFnF:TMons}
  Under DFS \<queue> and \<fail> cannot implement transaction
  monitors.
\end{lemma}

The main difference between these mechanisms and transaction
monitors is that the latter can execute functions without a 
contract being invoked at particular moments in the execution of transactions.
Take for example procedure \(\<init>\), neither \<queue> nor \<fail>
can simulate \(\<init>\), as there is no way for these mechanisms to
distinguish the first execution of a smart contract in a given
transaction.

Combining \<fail> with \<first> one can implement transaction monitors
in any execution order, including DFS (Theorem~\ref{thm:FstFnF:TMon}),
but \<fail> is not enough to implement transaction monitors in DFS.
Therefore, we conclude that DFS blockchains do not implement \<first>.
Moreover, putting all previous lemmas together, we conclude that a
DFS blockchains cannot implement any of the mechanisms listed in
Section~\ref{sec:exec:mechanisms} directly.

\begin{corollary}
  DFS blockchains cannot implement \<first>, \<fail>, \<ustore> or \<queue>.
\end{corollary}
All proofs are in Appendix~\ref{app:dfs}.


%% file: Discussion.tex
\section{Conclusion and Future Work}\label{sec:discussion}

We have studied transaction monitors for smart contracts.
Transaction monitors are a defense mechanism enabling smart contracts to
explicitly state wanted or unwanted behaviour at the transactional level.
This kind of properties are motivated by contracts like flash loans,
which are not implementable in their full generality in current
blockchains.
We propose a solution based on adding new mechanisms to the blockchain.
Transaction monitors can be incorporated directly into contracts or
simulated if some of these mechanisms are implemented.
This could be preferable since some of these mechanisms are very simple
and backward compatible, while others extend the functionality of
smart contracts.
We have studied how some mechanisms simulate each other, both for any
execution order, and specifically for BFS and DFS blockchains.
The conclusion is that the simplest mechanism that allows us to implement
transaction monitors is the combination of \<first> and \<fail>.

For simplicity, we have neglected a specific analysis of gas
consumption, except for recurrent operations that purposefully fail by
exhausting gas.
Even though transaction monitors will consume additional gas which can
influence the failure of the transaction (as with operation monitors),
we claim that for all our development there is an amount of gas that
can be calculated which will not make accepting transactions fail.
However, we leave a detailed study for future work.

Other avenues of future work include the study of new features,
particularly \emph{views} that allows contracts to inspect the state
of other contracts.
We are also performing a thorough study of how exposing new mechanisms
to contracts---that can use them for implementing functionality---can
break (or not) implementations of monitors that are correct without
adding the mechanisms.

%% file: Appendix/Mechanisms.Equiv.tex
\section{Mechanisms Equivalence}\label{app:equiv}

\begin{lemma}
    Mechanisms \first and \countmech are equivalent.
\end{lemma}
\begin{proof}
Assume \countmech, we implement \first by checking that its value is 0.

Assume \first.
For each smart contract \(A\) installed in the blockchain, we extend the storage
of \(A\) with a natural variable \(A_{\<count>}\), and we replace each method in
\(A\) with a new one where we first check \first: if it is true set
\(A_{\<count>} = 1\), if it not, add one to \(A_{\<count>}\).
\end{proof}

\begin{lemma}
    Mechanisms \first and \opmem are equivalent.
\end{lemma}
\begin{proof}
  Assume \opmem.
  For each smart contract \(A\) installed in the blockchain, we define a new
  volatile boolean value \(b_{\<init>}\).
  The variable \(b_{\<init>}\) is initialized to true and set to false at the
end of every method in \(A\).
Therefore, the value \(b_{\<init>}\) is true only at the first interaction
with \(A\) and false afterward.

Assume \first, let \(B\) be a smart contract and \(\Om\) be the volatile memory
segment of \(B\).
Extend \(B\) storage with \(\Om\)', a copy of \(\Om\).
At every method in \(B\), we check \first, if it is true we initialize \(\Om\)'
in the same way as the procedure from \opmem.
%

\end{proof}

\begin{lemma}\label{l:first:bsh}
    Mechanisms \first and \<bstore> are equivalent.
\end{lemma}
\begin{proof}
  Assume \first.
  Let \(A\) be a smart contract defining a \(hookup\) function updating the
storage of \(A\).
We expand the storage of \(A\) with a copy of the storage, \(s_{\<hookup>}\).
Moreover, at the end of every method in \(A\), we apply \(hookup\) with the
current storage and store its result at \(s_{\<hookup>}\).
Since function \(hookup\) cannot fail, it does not matter if it is executed at
the end of a transaction or at the very beginning of the next one.
Finally, we can check with \first and update the storage with the value stored
in \(s_{\<hookup>}\).

%
Note that if the contract \(A\) is never invoked again the last change will not
happen.
However, from the point of view of \(A\), it is equivalent.


Assume \<bstore>, and let \(B\) be a smart contract.
Extend \(B\) storage with a boolean field \(b_{\<fst>}\) to represent first, such
that it is initialized to \emph{true}, each invocation set it to \emph{false}
and \<bstore> set it back to \emph{true}.
\end{proof}


%% file: Appendix/BFS.ImpProof.tex
\section{Impossibility of Monitorability in BFS}
\label{app:bfs:imp}

\setcounter{backup}{\value{theorem}}
\setcounter{theorem}{\value{thm-bfs-monitors}}

\begin{theorem}
A BFS blockchain does not implement transaction monitors.
\end{theorem}

\begin{proof}
    Consider a transaction monitor for $A$ that fails smart contract $A$
    is called \textbf{exactly} once in a transaction.
    The monitor storage contains natural number, function \(\<init>\) sets
    the monitor storage to $0$, $\<begin>$ adds one to the counter,
    $\<end>$ does nothing, and $\<term>$ fails if the monitor storage is exactly
    one.

    Now, let $(\Sigma,\Delta)$ be a blockchain configuration, and let $A$
    and $B$ be two smart contracts, where $A$ is being monitored for the
    ``only once'' property. Consider the following family of
    executions of external operations originated by \(B\):

    \begin{itemize}
        \item operation $o$ where $B$ calls $A$ once,
        \item operation $o_k$  where $B$ calls to method \(f\) at smart contract $B$
      with parameter $k$ and also calls $A$. The method \(f\) calls itself recursively $k$ times and then it will call $A$.
    \end{itemize}

    When we execute previous operations in \((\Sigma,\Delta)\),
  we have the following traces:

    \[
    \begin{array}{rcl}
      \bullet & t : & (\Sigma,\Delta, [o]) \leadsto_{\<bfs>} (\Sigma_1,\Delta_1, [a_{1}]) \ldots\\
      \bullet & t_{0} : & (\Sigma,\Delta, [o_{0}]) \leadsto_{\<bfs>} (\Sigma_1,\Delta_1, [b_{0}, a_{1}]) \leadsto_{\<bfs>} (\Sigma_2,\Delta_2, [a_1,a_2] ) \leadsto_{\<bfs>} \\ 
              &         & (\Sigma_3,\Delta_3, [a_2] ++ as_1 ) \leadsto_{\<bfs>} (\Sigma_4,\Delta_4, as_1 ++ as_2 ) \ldots \\
      \bullet & t_{1} : & (\Sigma,\Delta, [o_{1}]) \leadsto_{\<bfs>} (\Sigma_1,\Delta_1, [b_{1},a_{1}]) \leadsto_{\<bfs>} (\Sigma_2,\Delta_2, [a_{1},b'_{1}]) \leadsto_{\<bfs>} \\ 
              &         & (\Sigma_3,\Delta_3, [b'_1] ++ as_1) \leadsto_{\<bfs>} (\Sigma_4,\Delta_4, as_1 ++ [a_2] ) \ldots \\
    \end{array}
  \]

    Monitor ``only once'' rejects transaction $t$, but
    accepts transactions \(t_k\), for all $k \geq 0$.
    Notice that $a_1$ is executed in the same blockchain configuration  in all transactions,
    therefore it must behave the same in all cases.
    The execution of $a_{1}$ in $t$ has three options:
    \begin{enumerate}
        \item do nothing/modify the storage,
        \item fail,
        \item generate a new invocation to $A$.
    \end{enumerate}
    In the first case $t$ will not satisfy the property ``only once'' as the transaction will not fail since $A$ never regains the control.
    On the other hand, in the second case all transactions $t_k$ will violate the monitor as all of them will fail while executing $a_{1}$.
    Therefore, $a_{1}$ must generate a new invocation to $A$, $op_a$. 
    When $op_a$ executes it will have the same options as $a_{1}$, and by making a similar analysis 
    as before, we can conclude that it must also generate a new invocation. 
    The only difference with $a_{1}$ is in the transaction $t_0$ 
    where $a_2$ has already executed and therefore $op_a$ can do nothing. 
    By repeating this argument, we have that the 
    operation generated by $a_{1}$ must be a \emph{recurring operation} (see Section~\ref{s:results:BFSScheduler})
    that keeps invoking itself until another method of $A$ is invoked. 
    If no other method of $A$ is invoked in the same transaction, the recurring operation
    will fail by gas exhaustion.
    This will work correctly for transactions $t$ and $t_k$ with $k \geq 0$.
    %
    %
    However, this will not work for all transactions.
    Let's consider a transaction similar to $t_0$ where $A$ is invoked three times.
    It is originated by an operation $o'_0$ where $B$ calls to method \(f\) at smart contract $B$
    with parameter $0$, then calls $A$ and finally calls method \(f\) with parameter $0$ again. 
    When we execute $o'_0$ in \((\Sigma,\Delta)\), we have the following trace:

    \[
    \begin{array}{rcl}
      \bullet & t'_{0} : & (\Sigma,\Delta, [o_{0}]) \leadsto_{\<bfs>} (\Sigma_1,\Delta_1, [b_{0}, a_{1},b_{0}]) \leadsto_{\<bfs>} (\Sigma_2,\Delta_2, [a_1,b_{0},a_2] ) \leadsto_{\<bfs>} \\ 
              &         & (\Sigma_3,\Delta_3, [b_0, a_2] ++ as_1 ) \leadsto_{\<bfs>} (\Sigma_4,\Delta_4, [a_2] ++ as_1 ++ [a_3] ) \ldots \\
    \end{array}
  \]
    Notice that the third invocation of $A$ originated by $B$ will execute after the recurring operation
    generated by $a_1$ stops.
    \footnote{if the recurring operation does not stop immediately after it see that other method in $A$ was invoked
     it must still stop after a finite number of steps and therefore we can still create a transaction where the third invocation to $A$ executes after it stops.}
    So, if we compare it with the case where $t$ is executed after $t_0$ we have that $a_3$ and $a_1$ will execute in the same state.
    Therefore, they must have the same behaviour.
    As explained before, when $a_{1}$ executes it must call the \emph{recurring operation}.
    However, $a_3$ is the last invocation of $A$ in $t'_0$, then its \emph{recurring operation}
    will keep calling itself until it fails by gas exhaustion, violating
    the property of monitor only once.\qed
\end{proof}

\setcounter{theorem}{\value{backup}}

\setcounter{backup}{\value{lemma}}
\setcounter{lemma}{\value{lemma-bfs-queue}}
  
\begin{lemma}
  In BFS \<ustore> cannot implement \<queue>.
\end{lemma}

\begin{proof}
  Assume a BFS blockchain that implements \<queue>.  Let $A$ be a smart
  contract that fails if and only if \<queue> is false.  Let
  \((\Sigma,\Delta)\) be a blockchain system, and consider the
  following two executions of external operations in
  \((\Sigma,\Delta)\):
  \begin{compactitem}
  \item operation $o_1$ calls $B.f$ which then generates two operations:
    $o_{A1}$ to $A$ and $o_{C1}$ a call to $C$ which does not generate
    any new operation.
  \item operation $o_2$, which calls $B.g$ which in turn generates
    $o_{A1}$, a call to $A$.
  \end{compactitem}
  Clearly, $A$ will fail in the execution of $o_1$ but not on $o_2$.
  In both cases the execution $o_{A1}$ will run in the same blockchain
  configuration and with same parameter and without \<queue> it is not
  possible for $A$ to know that there are pending operations as they
  will execute in another smart contract.
  Therefore, $o_{A1}$ will behave exactly the same in both transactions.
  As this is the only operation in $A$ in both transactions, the
  hookup method from the \<ustore> will also have the same behaviour
  in both cases.
  As consequence, either both transactions fail (making the
  transaction generated by $o_2$ fail incorrectly) or none (making the
  transaction generated by $o_1$ succeed incorrectly).
  Therefore, the smart contract $A$ cannot be implemented in a
  BFS blockchain equipped with \<ustore> but without \<queue>.\qed
\end{proof}

\setcounter{lemma}{\value{backup}}


%% file: Appendix/DFS.Proofs.tex
\section{DFS Proofs}
\label{app:dfs}

\setcounter{backup}{\value{lemma}}
\setcounter{lemma}{\value{lem-dfs-queue-first}}

\setcounter{lemma}{\value{backup}}

\setcounter{backup}{\value{lemma}}
\setcounter{lemma}{\value{lem-dfs-fail-tmons}}

\begin{lemma}
  \label{lmm:DFS:QInfoFnF:TMons:proof}
  Under DFS \<queue> and \<fail> cannot implement transaction
  monitors.
\end{lemma}

\begin{proof}
  Consider a transaction monitor for $A$ that fails smart contract $A$
  is called \textbf{exactly} once in a transaction.
  The monitor storage contains natural number, function \(\<init>\) sets
  the monitor storage to $0$, $\<begin>$ adds one to the counter,
  $\<end>$ does nothing, and $\<term>$ fails if the monitor storage is exactly
  one.

  Now, let $(\Sigma,\Delta)$ be a blockchain configuration, and let $A$,
  $B$ and $C$ be three smart contracts, where $A$ is being monitored for the
  ``only once'' property.
  %
  We analyze the pending queue of execution of two possible external operations
  originated by \(B\):
  \begin{enumerate}
    \item \(o_{1}\) where $B$ calls $A$ once and then $C$
    \item \(o_{2}\) where $B$ calls $A$ twice and then $C$
  \end{enumerate}
  We assume that there are no additional invocations to \(A\) aside from the
  described above.
  When we execute both operations in a \((\Sigma,\Delta)\) blockchain configuration,
  we have the following two traces:
  \[
    \begin{array}{rcl}
      \bullet & t_{1} : & (\Sigma,\Delta, [o_{1}]) \leadsto_{\<dfs>} (\Sigma',\Delta', [a_{1},c_{1}]) \ldots \\
      \bullet & t_{2} : & (\Sigma,\Delta, [o_{2}]) \leadsto_{\<dfs>} (\Sigma',\Delta', [a_{2},a_{3},c_{1}]) \ldots
    \end{array}
  \]
  
  Note that the presence of operation \(c_{1}\) in the pending execution queue is
  forcing the mechanism \QueueInfo to return false.
  Operations \(a_{1}\) and \(a_{2}\) are exactly the same operation invoking smart
  contract \(A\).
  Since operations \(o_{1}\) and \(o_{2}\) execute on the same blockchain configuration,
  i.e. \((\Sigma,\Delta)\), the result of executing \(a_{1}\) and \(a_{2}\) should be same
  on \(A\).
  
  The transaction executing \(o_{1}\) must fail because $A$ is called only
  once.
  Then, this execution of $A$ must fail, and thus, \(a_{1}\) should fail
  directly, set the failing bit to true, or call another operation that will fail later.
  
  The transaction executing \(o_{2}\) should not fail.
  Therefore, \(a_{1}\) cannot fail directly, and thus, it must either 
  set the failing bit to true or call another operation, \(a_{4}\).
  In the latter case, since the scheduler strategy is DFS, 
  operation \(a_{4}\) replaces \(a_{1}\) as
  the head of the pending execution queue in both cases.
  Moreover, every operation resulting from the execution of operation \(a_{4}\) is
  going to be appended to the head of the pending execution queue.
  
  \[
    \begin{array}{rcl}
      \bullet & t_{1} : & (\Sigma,\Delta, [o_{1}]) \leadsto_{\<dfs>} (\Sigma',\Delta', [a_{1},c_{1}]) \leadsto_{\<dfs>} (\Sigma'',\Delta'', [a_{4},c_{1}]) \ldots\\
      \bullet & t_{2} : & (\Sigma,\Delta, [o_{2}]) \leadsto_{\<dfs>} (\Sigma',\Delta', [a_{2},a_{3},c_{1}]) \leadsto_{\<dfs>} (\Sigma'',\Delta'', [a_{4},a_{3},c_{1}]) \ldots
    \end{array}
  \]

  If operations \(a_{1}, a_{4}\) or any of its decedents explicitly fail, then the execution
  of operation \(o_{2}\) is doomed to fail too.
  
  However, if none of these operations, \(a_{1},a_{4},\ldots\), fail then the first run
  will finish without failing, violating the property of monitor only once.

  Therefore, \(a_{1}\) or any of its decedents must set the failing bit to true.
  This will work correctly for the execution of operation \(o_{1}\).
  In order for this to also work for the execution of operation \(o_{2}\)
  it is enough if \(a_{3}\) sets the failing to false.
  However, if we consider the transaction that begins with two external operations  
  \([o_{2};o_{1}]\) and the execution of the transaction $t_2$ followed by $t_1$.
  We have that $o_1$ will run in both cases in the same blockchain 
  configuration, and therefore it will have the same behaviour in both cases. 
  As explained before, when \(a_{1}\) executes it or one of it descendants must 
  set the failing bit to true.
  This implies that the execution originated by \([o_{2};o_{1}]\) will finish with $A$'s
  failing bit set to true and thus the whole transaction will fail, violating
  the property of monitor only once.\qed
\end{proof}

\setcounter{lemma}{\value{backup}}

\begin{lemma}\label{lmm:DFS:Not:QInfo}
  A DFS blockchain does not implement \<queue>.
\end{lemma}
\begin{proof}
Let \(\mathcal{B}\) be a DFS blockchain and \(A\) a smart contract installed
in \(\mathcal{B}\) that fails if and only if \<queue> is false.  
Let \((\Sigma,\Delta)\) be a blockchain system, and consider the
following two executions of external operations in
\((\Sigma,\Delta)\):
\begin{compactitem}
\item operation $o_1$ calls $B.f$ which then generates two operations:
  $o_{A1}$ to $A$ and $o_{C1}$ a call to $C$ which does not generate
  any new operation.
\item operation $o_2$, which calls $B.g$ which in turn generates
  $o_{A1}$, a call to $A$.
\end{compactitem}
  Clearly, $A$ will fail in the execution of $o_1$ but not on $o_2$.
  In both cases the execution $o_{A1}$ will run in the same blockchain
  configuration and with same parameter and without \<queue> it is not
  possible for $A$ to know that there are pending operations as they
  will execute in another smart contract.
  Therefore, $o_{A1}$ will behave exactly the same in both transactions.
  As consequence, either both transactions fail (making the
  transaction generated by $o_2$ fail incorrectly) or none (making the
  transaction generated by $o_1$ succeed incorrectly).
  Therefore, the smart contract $A$ cannot be implemented in a
  DFS blockchain without \<queue>.\qed
\end{proof}



%% file: Appendix/FlashLoans.tex
\section{Implementations of Flash Loan}
\input{Contracts/Flashloan_ustoragehookup.tex}
\input{Contracts/Flashloan_first_fnfhookup.tex}
\input{Contracts/Flashloan_first_bfs.tex}
\input{Contracts/Flashloan_bfs_queueinfo.tex}

%% file: Contracts/Flashloan_ustoragehookup.tex
\begin{figure}[t!]
    \begin{lstlisting}[language=Solidity,numbers=none,caption=A correct Flash Loan implementation using U. Storage Hookup,label={flashLoan:lenderWithUStorageHookup}]
      contract Lender {
        uint initial_balance;    
        function lend(address payable dest, uint amount) public {
          require(amount <= this.balance);
          dest.transfer(amount);
        }  
        ustore { 
            assert(this.balance >= initial_balance); 
            initial_balance = this.balance;
        }
      }
   \end{lstlisting}
\end{figure}

%% file: Contracts/Flashloan_first_fnfhookup.tex
\begin{figure}[t!]
    \begin{lstlisting}[language=Solidity,numbers=none,caption=A correct Flash Loan implementation using First + Fail/NoFail Hookup,label={flashLoan:lenderWithFirstFNFHookup}]
      contract Lender {
        uint initial_balance;    
        function lend(address payable dest, uint amount) public {
          if(this.first){
             initial_balance = this.balance; 
          }  
          require(amount <= this.balance);
          dest.transfer(amount);
          check_balance();
        }  
        receive() external payable {
          check_balance();
        }
        function check_balance() private {
            if(this.balance < initial_balance){
              this.fail = true;
            } else {
              this.fail = false;
            }
        }
      }   
   \end{lstlisting}
\end{figure}

%% file: Contracts/Flashloan_first_bfs.tex
\begin{figure}[t!]
    \begin{lstlisting}[language=Solidity,numbers=none,caption=A correct Flash Loan implementation using BFS Scheduler and First. Each function returns the list of generated operations.,label={flashLoan:lenderWithFirstBFS}]
      contract Lender {
        uint initial_balance;    
        function lend(address payable dest, uint amount) public {
          if(this.first){
             initial_balance = this.balance; 
          }  
          return [system.gen_op(dest, (), amount),system.gen_op(this.check_balance, (), 0)];  
        }  
        receive() external payable {
          return [system.gen_op(this.check_balance, (), 0)];
        }
        function check_balance() private { 
            if(this.balance < initial_balance){
                return [system.gen_op(this.check_balance, (), 0)];
            }
            return [];
        }
      }   
   \end{lstlisting}
\end{figure}

%% file: Contracts/Flashloan_bfs_queueinfo.tex
\begin{figure}[t!]
    \begin{lstlisting}[language=Solidity,numbers=none,caption=A correct Flash Loan implementation using BFS + Qeueu Info,label={flashLoan:lenderWithBFSQueueInfo}]
      contract Lender {
        uint initial_balance;    
        function lend(address payable dest, uint amount) public {  
          require(amount <= this.balance);
          dest.transfer(amount);
          check_balance();
        }  
        receive() external payable {
          check_balance();
        }
        function check_balance() private {
            if(system.queue) {
                assert(this.balance >= initial_balance);
                initial_balance = this.balance;
            } else {
                check_balance();
            }
        }
      }   
   \end{lstlisting}
\end{figure}